\documentclass[conference]{IEEEtran}
\usepackage{graphicx}
\usepackage{amsmath}
\usepackage{amssymb}
\usepackage{algorithm}
\usepackage{algpseudocode}

\newcommand{\Fig}[1]{Fig.~\textup{\ref{#1}}}
\graphicspath{{figures/}}

\newtheorem{lemma}{Lemma}
\newtheorem{theorem}{Theorem}

\newtheorem{remark}{Remark}

\newtheorem{proposition}{Proposition}

\newcommand{\F}{\mathbb{F}}
\newcommand{\rp}{\rho}

\begin{document}

\title{An Upper Bound on the Minimum Distance of LDPC Codes over GF(q)}

\author{
  \IEEEauthorblockN{Alexey Frolov}
	
  \IEEEauthorblockA{\small Inst. for Information Transmission Problems\\
    Russian Academy of Sciences\\Moscow, Russia\\
    Email: alexey.frolov@iitp.ru
  }
}


%


\maketitle

\begin{abstract}
In \cite{BHL} a syndrome counting based upper bound on the minimum distance of regular binary LDPC codes is given. In this paper we extend the bound to the case of irregular and generalized LDPC codes over $GF(q)$. The comparison to the lower bound for LDPC codes over $GF(q)$ and to the upper bound for non-binary codes is done. The new bound is shown to lie under the Gilbert--Varshamov bound at high rates.
\end{abstract}

\section{Introduction}
In this paper we investigate the minimum code distance of LDPC codes \cite{G, T} over $\F_q$. Such codes have good error-correcting capabilities, efficient encoding and decoding algorithms. All of these makes the codes very popular in practical applications.

In \cite{BHL} a syndrome counting based upper bound on the minimum distance of regular binary LDPC codes is given. In this paper we extend the bound to the case of irregular and generalized LDPC codes over $\F_q$. 

Our contribution is as follows. First we derive the upper bound for generalized LDPC codes (we assume the Tanner graph \cite{T} to be a regular one) over $\F_q$. The bound depends on the weight\footnote{Here and in what follows by weight we mean the Hamming weight, i.e. a number of non-zero elements in a vector.} enumerator of the constituent code. Second we derive the upper bound for irregular LDPC codes (we assume the Tanner graph to be an irregular one) over $\F_q$. The constituent code in this case is a single parity-check (SPC) code over $\F_q$. We compare the new upper bound to the lower bound for LDPC codes over $\F_q$ \cite{FZ1} and to the upper bound for non-binary codes \cite{BHL1}. At last we show the derived bound to lie under the Gilbert--Varshamov bound at high rates.

\section{Generalized LDPC codes}

In this section we obtain the upper bound on the minimum distance of generalized LDPC codes. We use Elias--Bassalygo type arguments \cite{BE}.

Let us briefly consider the construction of generalized LDPC code $\mathcal{C}$ over $\F_q$. To construct such a code we use a bipartite graph, which is called the Tanner graph \cite{T} (see \Fig{tanner}). The graph consists of $N$ variable nodes $v_1, v_2, \ldots, v_N$ and $M$ check nodes $c_1, c_2, \ldots, c_M$. In this section we assume all the check nodes to have the same degree $n_0$ (such Tanner graphs are called right regular ones). We associate constituent codes to each of the check nodes. In this section all the constituent codes are the same (we denote the constituent code by $\mathcal{C}_0$). We assume $\mathcal{C}_0$ to be a linear $[n_0, R_0, d_0]$-code over $\F_q$. Let us denote the parity-check matrix of the constituent code by $\mathbf{H}_0$. The matrix has size $m_0 \times n_0$, where $m_0 = (1-R_0)n_0$.

\begin{figure}[t]
\centering
\includegraphics[width=0.4\textwidth]{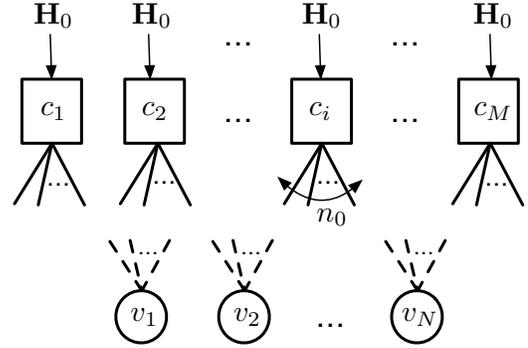}
\caption{Tanner graph}
\label{tanner}
\end{figure}

Let $G(s, n_0, d_0)$ be the weight enumerator of the code $\mathcal{C}_0$, i.e.
\[
G(s, n_0, d_0) = 1+\sum\limits_{i=d_0}^{n_0} A(i) s^i,
\]
where $A(i)$ is  the number of codewords of weight $i$ in a code $\mathcal{C}_0$.

To check if $\mathbf{r} = (r_1, r_2, \ldots, r_N) \in \F_q^N$ is a codeword of $\mathcal{C}$ we associate the symbols of $\mathbf{r}$ to the variable nodes ($v_i = r_i, i = 1,\ldots, N$). The word $\mathbf{r}$ is called a codeword of $\mathcal{C}$ if all the constituent codes are satisfied (the symbols which come to the codes via the edges of the Tanner graph form codewords of the constituent codes).

It is clear the resulting code $\mathcal{C}$ is linear, so it has a parity-check matrix associated to it. We denote the matrix by $\mathbf{H}$. The code is over $\F_q$ and has the length $N$.
 
By  $\mathbf{S}$ we denote the resulting syndrome of a generalized LDPC code, i.e, for a received sequence $\mathbf{r}$
\[
\mathbf{S} = \mathbf{H} \mathbf{r}^T.
\]
The syndrome consists of the constituent code syndromes and can be presented in such a way
\[
\mathbf{S} = (\mathbf{S}_1, \mathbf{S}_2, \ldots, \mathbf{S}_M),
\]
where $\mathbf{S}_i$, $i = 1, \ldots, M$, is a syndrome of the $i$-th constituent code. 

Let us introduce a notation. For a discrete random variable $X$, $H_Q(X)$ denotes the entropy of $X$, i.e.,
\[
H_Q(X) = -\sum\limits_{x} \Pr(X=x) \log_Q\Pr(X=x).
\]

In what follows we will need the fact formulated in the Lemma below
\begin{lemma}\label{lemma::entropy}
Let $X$ be the random variable taking $t$ values, let $p^* \geq 1/t$ and let 
\[
\Pr(X = x_i) = p_i \leq p^*, \:\: \forall i=1, \ldots, t,
\]
then
\[
H_Q(X) \geq -\log_Q(p^*).
\]
\end{lemma}

Let us introduce some additional notation. For a real number $0 \leq x \leq 1$ let
\[
h_Q(x) = -x \log_Q x - (1-x) \log_Q (1-x) + x \log_Q(Q-1).
\]
be $Q$-ary entropy function.

We are ready to prove a theorem
\begin{theorem}
Let $\mathcal{C}$ be a generalized LDPC code of length $N$, rate $R$, minimum distance $\delta N$, with constituent $[n_0, R_0, d_0]$ code $\mathcal{C}_0$ over $\F_q$. Let $G(s, n_0, d_0)$ be the weight enumerator of $\mathcal{C}_0$. Then for sufficiently large $N$ the following inequality holds
\[
R(\mathcal{C}) \leq 1 - \max\limits_{\delta/2 \leq \omega \leq 1}\left[ \frac{h_q(\omega) - R_{CW}(q, \omega, \delta)}{h_{q^{m_0}}(1-p_0(\omega))} \right] + o(1),
\]
where
\[
p_0(\omega) = (1-\omega)^{n_0} G\left(\frac{\omega}{(1-\omega)(q-1)}, n_0, d_0\right).
\]
\end{theorem}
\begin{proof}
Consider all possible vectors of length $N$, weight $W = \omega N$ over $\F_q$. We introduce an equiprobable distribution on such vectors.  Recall, that $\mathbf{S}$ denotes the syndrome and $\mathbf{S}_i$, $i=1,\ldots, M$, denotes the syndrome of the $i$-th constituent code. $\mathbf{S}$ and $\mathbf{S}_i$, $i=1,\ldots, M$, are random variables. 

Note, that
\begin{equation}\label{main_ineq}
H_q(\mathbf{S} = (\mathbf{S}_1, \mathbf{S}_2, \ldots, \mathbf{S}_M)) \leq \sum\limits_{i=1}^{M} H_q(\mathbf{S}_i). 
\end{equation}

Our aim now is to estimate left and right parts of the inequality (\ref{main_ineq}). 

Let us start from the \textit{\bf left part} of (\ref{main_ineq}). Let us consider the probability $\Pr(\mathbf{S} = \mathbf{s})$ for some fixed syndrome $\mathbf{s}$. It is clear that the number of vectors of weight $\omega N$ giving the syndrome $\mathbf{s}$ is upper bounded with the maximal cardinality $B(q, \omega N, \delta N)$ of a constant weight code with distance $\delta N$ over $\F_q$, in other words
\[
\Pr(\mathbf{S} = \mathbf{s}) \leq \frac{B(q, \omega N, \delta N)}{\binom{N}{\omega N} (q-1)^{\omega N}}.
\]

After applying Lemma~\ref{lemma::entropy} we have
\begin{eqnarray}\label{left_part}
H_q(\mathbf{S}) &\geq& -\log_q \left( \frac{B(q, \omega N, \delta N)}{\binom{N}{\omega N} (q-1)^{\omega N}} \right) \nonumber\\
                           &\geq& N ( h_q(\omega) - R_{CW}(q, \omega, \delta) + o(1) ),
\end{eqnarray}
where $R_{CW}(q, \omega, \delta)$ is an upper bound of the rate of constant weight code. 

Now we proceed with the \textit{\bf right part} of (\ref{main_ineq}). Let us consider the $i$-th constituent code, recall, that $\mathbf{S}_i$ is a random variable and it is easy to see that 
\begin{flalign*}
&p_0 = \Pr(\mathbf{S}_i = \mathbf{0}) \\
&= \frac{1}{\binom{N}{W} (q-1)^W} \left[\sum\limits_{i=0}^{n_0} \left\{ A(i) \binom{N-n_0}{W-i} (q-1)^{W-i}\right\} \right].
\end{flalign*}

We are interesting in asymptotic estimate when $N \to \infty$. In this case we have
\[
\frac{\binom{N-n_0}{W-i}}{\binom{N}{W}} \to \omega^i (1-\omega)^{n_0-i}
\]
and
\begin{eqnarray*}
p_0 &=& \left[\sum\limits_{i=0}^{n_0} \left\{ A(i) \omega^i (1-\omega)^{n_0-i} (q-1)^{-i}\right\} \right] + o(1) \\
       &=& (1-\omega)^{n_0} G\left(\frac{\omega}{(1-\omega)(q-1)}, n_0, d_0\right) + o(1). 
\end{eqnarray*}

After applying the log sum inequality for the entropy of the random variable $\mathbf{S}_i$ we have
\begin{flalign}\label{right_part}
&H_{q}(\mathbf{S}_i) = -\sum\limits_{j = 0}^{q^{m_0}-1} \Pr(\mathbf{S}_i = \mathbf{s}_j) \log_{q} \Pr(\mathbf{S}_i = \mathbf{s}_ j) \nonumber\\
&= -p_0 \log_{q} p_0 - \sum\limits_{j = 1}^{q^{m_0}-1} \Pr(\mathbf{S}_i = \mathbf{s}_j) \log_{q} \Pr(\mathbf{S}_i = \mathbf{s}_ j)  \nonumber\\     
&\leq -p_0 \log_{q} p_0 - (1-p_0) \log_{q} \frac{1-p_0}{q^{m_0}-1} \nonumber\\
& = m_0 h_{q^{m_0}}(1-p_0). 
\end{flalign}

Finally after substituting of (\ref{left_part}) and (\ref{right_part}) into (\ref{main_ineq}) we obtain
\begin{equation}\label{subst_omega}
R \leq 1 - \frac{h_q(\omega) - R_{CW}(q, \omega, \delta)}{h_{q^{m_0}}(1-p_0(\omega))} + o(1).
\end{equation}

Now the maximization domain is $0 < \omega \leq 1$, to finish the proof we need to reduce it to $ \delta/2 < \omega \leq 1$. We just need to note, that for $\omega \leq \delta/2$
\[
R_{CW}(q, \omega, \delta) = 0
\] 
and maximum (for this sub-interval) is achieved at $\omega = \delta/2$.
\end{proof}

\section{Irregular LDPC codes}
In this section we derive the upper bound for irregular LDPC codes over $\F_q$. We assume the Tanner graph to be irregular. The constituent code in this case is a single parity-check (SPC) code over $\F_q$.

First we note that an SPC code over $\F_q$ is an MDS code. For the MDS code the number of codewords of weight $W$ can be calculated as follows 
\begin{flalign*}
&A(W) = [s^W]G(s, d_0, n_0) \\
&= \binom{n_0}{W} (q-1) \sum\limits_{j=0}^{W-d_0} \left\{ (-1)^j \binom{W-1}{j} q^{W-d_0-j}\right\}.
\end{flalign*}

Thus the enumerator of an SPC code over $\F_q$ is as follows
\[
G(s, d_0=2, n_0) = \frac{1}{q} \left( 1 + (q-1)s\right)^{n_0} + \frac{q-1}{q} (1-s)^{n_0}.
\]

To formulate a theorem we need a notion of row degree polynomial
\[
\rp(x) = \sum\limits_{i=r_{\min}}^{r_{\max}} \rp_i x^i,
\]
where $\rp_i$ is a fraction of rows of the parity check matrix of weight $i$, $r_{\min}$ and $r_{\max}$ are the minimal and maximal row weights accordingly.

\begin{theorem}
Let $\mathcal{C}$ be an LDPC code of length $N$, rate $R$, minimum distance $\delta N$, with row degree polynomial $\rp(x)$. Then for sufficiently large $N$ the following inequality holds
\begin{flalign*}
&R(\mathcal{C}) \leq \overline{R}(q, \rp(x)) \\
&= 1 - \max\limits_{\delta/2 \leq \omega \leq 1}\frac{h_q(\omega) - R_{CW}(q, \omega, \delta) }{h_{q}\left[ \frac{q-1}{q} \left(1 - \rp \left( 1 - \frac{q}{q-1} \omega \right) \right)\right] } + o(1).
\end{flalign*}
\end{theorem}
\begin{proof}
Consider the right part of (\ref{main_ineq}), we have
\begin{flalign*}
&\frac{1}{N}\sum\limits_{i=1}^{M} H_q(\mathbf{S}_i) \\
&= (1-R) \sum\limits_{i=r_{\min}}^{r_{\max}} \rp_i h_{q}\left[ 1-(1-\omega)^{n_0} G\left(\frac{\omega}{(1-\omega)(q-1)}\right)\right] \\
&= (1-R) \sum\limits_{i=r_{\min}}^{r_{\max}} \rp_i h_{q}\left[ \frac{q-1}{q} - \frac{q-1}{q} \left( 1 - \frac{q}{q-1} \omega \right)^{i}\right] \\
&\leq (1-R) h_{q}\left[ \frac{q-1}{q} - \frac{q-1}{q} \rp \left( 1 - \frac{q}{q-1} \omega \right) \right]. 
\end{flalign*}
These completes the proof.
\end{proof}

\begin{remark}
We note that the bound improves the result from \textup{\cite{BHL}} for the binary case. Recall that in \textup{\cite{BHL}} in case of irregular LDPC code it is suggested to just substitute $r_{\max}$ to the bound for regular code.
\end{remark}

At last we prove that the upper bound is better for regular codes (with the same average row degree as irregular codes).

\begin{proposition}\label{prop}
Let $b > 0$ be an integer, let $\rp(x)$ be the row degree distribution of irregular code, such that $\sum\nolimits_{i=r_{\min}}^{r_{\max}} i \rp_i = b$
and let $\rp_{\text{reg}} = x^{b}$, then
\[
\overline{R}(q, \rp(x)) \leq \overline{R}(q, \rp_{\text{reg}}(x)).
\]
\end{proposition}
\begin{proof}
Let $\alpha > 0$. By the concavity of the function $\alpha^x$ we have
\[
\rp(\alpha) \geq \alpha^{\sum\nolimits_{i=r_{\min}}^{r_{\max}} i \rp_i } = \rp_{\text{reg}}(\alpha).
\] 
These completes the proof.
\end{proof}

\section{Numerical results}

In this section we present the numerical results. We use an upper bound derived in \cite{BHL1} as a function $R_{CW}(q, \omega, \delta)$. To the best knowledge of the author the bound is currently the best upper bound on the rate of non-binary constant weight codes. The results are shown in Tables \ref{tab::irreg}, \ref{t_8} and \ref{t_64}. We use the following notation:

\begin{itemize}
\item $\delta_{GV}$ -- the Gilbert--Varshamov bound;
\item $\delta^{(U)}_{LDPC}$ -- the new bound for LDPC codes derived in the paper;
\item $\delta^{(L)}_{LDPC}$ -- the lower bound for LDPC codes from \cite{FZ1};
\item $\delta_{BHL}$ -- the upper bound on the minimum distance of non-binary codes \cite{BHL1}, which is an improvement of the Aaltonen bound \cite{A}. 
\end{itemize}

We first compare the values of the new estimate $\delta^{(U)}_{LDPC}$ for regular and irregular codes. In Proposition~\ref{prop} we proved that the bound is better for regular codes. Here we present some values calculated for 
$q = 8$ and $R = 0.9$. We fix the degree of the variable node $\ell = 3$. The results are shown in Table~\ref{tab::irreg}. We note, that for this case $\delta_{BHL} = 0.0638$ and $\delta_{GV} = 0.0328$.

\begin{table}
\caption{Comparison of regular and irregular LDPC codes for $q = 8$, $R = 0.9$}
\label{tab::irreg}
\centering
\begin{tabular}{|c|c|c|c|c|c|c|}
\hline
${\rp}_{15}$ & 0 & 0.25  & 0.125     & 0     \\
${\rp}_{20}$ & 0 & 0       & 0.125   & 0     \\
${\rp}_{25}$ & 0 & 0       & 0       & 0.5  \\
${\rp}_{30}$ & 1 & 0.5    & 0.5      & 0      \\
${\rp}_{35}$ & 0 & 0       & 0       & 0.5 \\
${\rp}_{40}$ & 0 & 0       & 0.125   & 0     \\
${\rp}_{45}$ & 0 & 0.25  & 0.125     & 0     \\
\hline
$\delta^{(U)}_{LDPC}$ & 0.0512 &  0.0493 & 0.0500 & 0.0512 \\
\hline
\end{tabular}
\end{table}

For now let us compare $\delta_{GV}$, $\delta_{up}$ and $\delta_{BHL}$ for the case of high-rate LDPC codes over $\F_8$. In Table~\ref{t_8} the results are shown. We choose regular $(\ell = 3, n_0)$ LDPC codes. We see that the new bound improves the best upper bound for non-binary codes ($\delta_{BHL}$). We also see that at very high rates ($R> 0.994$) the new bound lies below the Gilbert--Varshamov bound. We note that the interval of rates in which we observe this behavior is decreasing when $q$ grows. For $q = 2$ the interval is $R> 0.985$, for $q = 16$ the interval is $R > 0.997$.

\begin{table}
\caption{Results for high-rate codes, $q=8$}
\label{t_8}
\centering
\begin{tabular}{|c||c|c|c|c|c|c|}
\hline
$(\ell, n_0); R$ & $\delta_{GV}$ & $\delta^{(U)}_{LDPC}$ & $\delta_{BHL}$ \\
\hline
\hline
(3,10); 0.7       & 0.1260        &  0.2102 & 0.2239 \\
\hline
(3,50); 0.94     & 0.0179        &  0.0263 & 0.0355 \\
\hline
(3,100); 0.97   & 0.0080        &  0.0106 & 0.0106 \\
\hline
(3,200); 0.985 & 0.0036        &  0.0043 & 0.0073 \\
\hline
(3,500); 0.994 & {\bf 0.0013} & {\bf 0.0013} & 0.0026 \\
\hline
(3,600); 0.995 & {\bf 0.0011} & {\bf 0.0010} & 0.0021\\
\hline
\end{tabular}
\end{table}

At last we compare the new upper bound to the lower bound on the minimum distance of LDPC codes over $\F_q$. We use the lower bound from \cite{FZ1}. The results for $q = 64$ are shown in Table~\ref{t_64}.

\begin{table}[!t]
\caption{Comparison to the lower bound, $q=64$}
\label{t_64}
\centering
\begin{tabular}{|c||c|c|c|}
\hline
$(\ell, n_0); R$  &  $\delta_{GV}$ & $\delta^{(L)}_{LDPC}$ & $\delta^{(U)}_{LDPC}$\\
\hline
\hline
(14, 16); 0.125 & 0.7400 & 0.7355 & 0.8539  \\
\hline
(9, 12); 0.25 & 0.5894 & 0.5860 & 0.7319 \\
\hline
(15, 24); 0.375 & 0.4608 & 0.4585 & 0.6101 \\
\hline
(14, 28); 0.5 & 0.3462 & 0.3445 & 0.4881\\
\hline
(15, 40); 0.625 & 0.2427 & 0.2415 & 0.3661\\
\hline
(13, 52); 0.75 & 0.1492 & 0.1480 & 0.2441\\
\hline
(8, 64); 0.875 & 0.0665 & 0.0575 & 0.1221\\
\hline
\end{tabular}
\end{table}

\section{Conclusion}
The new upper bound on the minimum distance of generalized and irregular LDPC codes over $\F_q$ is derived. For the derivation of the bound we used Bassalygo--Elias type arguments. The bound is proved to be better for regular LDPC codes over $\F_q$. We compared the new upper bound to the lower bound for LDPC codes over $\F_q$ and to the upper bound for non-binary codes. We showed, that at very high rates ($R> 0.994$ for $q=8$) the new bound lies below the Gilbert--Varshamov bound. We note that the interval of rates in which we observe this behavior is decreasing when $q$ grows. For $q = 2$ the interval is $R> 0.985$, for $q = 16$ the interval is $R > 0.997$.

\section*{Acknowledgment}
The author thanks V.V. Zyablov for the numerous advice and recommendations. This work was partially supported by Russian Science Foundation grant 14-50-00150.


\newpage
\begin{thebibliography}{1}

\bibitem{BHL}
Y.~Ben-Haim and S.~Litsyn.
\newblock Upper Bounds on the Rate of LDPC Codes as a Function of Minimum Distance.
\newblock \emph{IEEE Trans. Inf. Theory},
vol. 52, no. 5, pp. 2092--2100, May 2006.

\bibitem{G}
R.~G.~Gallager, 
\newblock \emph{Low-Density Parity-Check Codes}.
\newblock Cambridge: MIT Press, 1963.

\bibitem{T}
R.~Tanner.
\newblock A recursive approach to low complexity codes.
\newblock \emph{IEEE Trans. Inf. Theory}, 
vol. 27, no. 5, pp. 533--547, Sep. 1981.

\bibitem{FZ1}
A. Frolov and V. Zyablov.
\newblock Bounds on the minimum code distance for nonbinary codes based on bipartite graphs.
\newblock \emph{Probl. Inf. Transm.},
vol. 47, no. 4, pp. 327--341, 2011.

\bibitem{BHL1}
Y.~Ben-Haim and S.~Litsyn.
\newblock A New Upper Bound on the Rate of Non-Binary Codes.
\newblock In Proc. IEEE Int. Symp. Inf. Theory, 9--14 July 2006, pp. 297--301.

\bibitem{BE}  
L.~A.~Bassalygo.
\newblock New Upper Bounds for Error Correcting Codes.
\newblock \emph{Problems Inf. Transm.},
vol. 1, no. 4, pp. 32--35, 1965.

\bibitem{A}
M. J. Aaltonen
\newblock A new upper bound on nonbinary block codes.
\newblock \emph{Discrete Mathematics}
vol. 83, pp. 139--160, 1990.


\end{thebibliography}
\end{document}